\documentclass[%
 reprint,
superscriptaddress,
 amsmath,amssymb,
 aps,
prl,
]{revtex4-1}

\usepackage{amsthm}
\usepackage{graphicx}% Include figure files
\usepackage{dcolumn}% Align table columns on decimal point
\usepackage{bm}% bold math
\usepackage{multirow} % multirow

\theoremstyle{remark}

\newtheorem{theorem}{\indent \emph{\textbf{Theorem}}}

\begin{document}

\makeatletter
\newcommand{\rmnum}[1]{\romannumeral #1}
\newcommand{\Rmnum}[1]{\expandafter\@slowromancap\romannumeral #1@}
\makeatother

\preprint{APS/123-QED}

\title{Coherence of Superpositions}% Force line breaks with \\
%\thanks{A footnote to the article title}
\author{Qiu-Ling Yue}%
\affiliation{State Key Laboratory of Networking and Switching Technology, Beijing University of Posts and Telecommunications, Beijing, 100876, China}
\author{Feng Liu}
\affiliation{State Key Laboratory of Networking and Switching Technology, Beijing University of Posts and Telecommunications, Beijing, 100876, China}
\affiliation{School of Mathematics and Statistics Science, Ludong University, Yantai 264025, China}
\author{Chao-Hua Yu}
\affiliation{State Key Laboratory of Networking and Switching Technology, Beijing University of Posts and Telecommunications, Beijing, 100876, China}
\author{Xiao-Li Wang}
\affiliation{State Key Laboratory of Networking and Switching Technology, Beijing University of Posts and Telecommunications, Beijing, 100876, China}
\author{Fei Gao}
 \email{gaofei\_bupt@hotmail.com}
\affiliation{State Key Laboratory of Networking and Switching Technology, Beijing University of Posts and Telecommunications, Beijing, 100876, China}
\author{Qiao-Yan Wen}
\affiliation{State Key Laboratory of Networking and Switching Technology, Beijing University of Posts and Telecommunications, Beijing, 100876, China}

\date{\today}% It is always \today, today,

\begin{abstract}
Given a pure state and an arbitrary  decomposition of it as a superposition of
two terms, is there a relation between the coherence of the superposition state and that of the two terms being superposed? Using the relative entropy of coherence, we give an affirmative answer to this question, and obtain the upper bounds of the coherence of superposition in terms of the coherence of two terms under different conditions:
1) the two terms have support on orthogonal subspaces;
2) the  two terms are  in the same subspace.
Furthermore, we also obtain the lower bound relation of coherence of superpositions.
%1) When the two terms have support on orthogonal subspaces, the difference between the coherence of their superposition state and the average coherence of them is smaller than 1 bit;
%2) When two terms are included in the same subspace, there is a  more complex relationship between them.
%Furthermore, we also obtain the lower bound of coherence of superpositions.
\begin{description}
\item[PACS numbers]
03.65.Aa, 03.65.Ta, 03.67.Mn
\end{description}
\end{abstract}

\pacs{Valid PACS appear here}% PACS, the Physics and Astronomy
\maketitle

%\tableofcontents

%\section{Introduction}

\emph{Introduction.---}Quantum coherence is a fundamental part in quantum mechanics.
%and it is the important reason for the quantum mechanics different from classical mechanics.
It  is also the  potential  significant physical resource in quantum information theory \cite{1}, and thermodynamics \cite{2,3}.
Recently, Baumgratz et al. \cite{duliang1} proposed a rigorous framework to quantify coherence. After that,
there are prolific researches on quantum coherence,
including different coherence measures \cite{duliang2,duliang3,duliang4},  the properties of coherence  \cite{xingzhi1,xingzhi2}.
the freezing phenomenon of coherence \cite{dongjie1},  the relation between coherence, entanglement and quantum correlation \cite{zhuanhua1,zhuanhua2,zhuanhua3},
and some other topics \cite{qita1,qita2,qita3,qita4}.

%Superposition  is a fundamental principle of quantum mechanics \cite{22}.
It's well-known that quantum coherence and quantum entanglement are both rooted in the superposition principle.
Naturally, we wonder how superposition affects them.
In the case of entanglement, Linden et al. \cite{jiuchan2}  raised a problem: what is the relation between the entanglement of a bipartite state and that of the two terms in the superposition? They found upper bounds on the entanglement of the superposition state in terms of the entanglement of the states being superposed. Since then, there are extensive works to further study entanglement of superpositions \cite{jiuchan3,jiuchan4,jiuchan5}.

In the case of coherence, it is also interesting to  consider the analogical problem: is there a relation between the coherence of superposition and that of two terms being superposed? That is,
given  an arbitrary state $ | \Omega  \rangle$, and a decomposition of it
\begin{equation}
 | \Omega  \rangle {\rm{ = }}\alpha | \Phi  \rangle + \beta | \Psi  \rangle,
\end{equation}
%with $|\alpha|^{2}+|\beta|^{2}=1$,
what is the relation between the coherence of $ | \Omega  \rangle$ and that of $ | \Phi  \rangle $ and $ | \Psi  \rangle$?

Before undertaking our study, it is worth of making some observations.
We consider the coherence of the following superposition under computation basis
\begin{equation}
 | {\Omega _ 1 } \rangle {\rm{ = }}\frac{{\rm{1}}}{{\sqrt {\rm{2}} }} ( { | 0 \rangle + | 1 \rangle } ).
\end{equation}
As we know,  $ | 0 \rangle$, $| 1 \rangle$ both are incoherent states, but their superposition $ | {\Omega_1}\rangle$ is a maximally coherent state.
%$ | 00 \rangle$, $| 11 \rangle$ both are incoherent states, but their superposition $ | {\Omega _2}\rangle$ is a coherent state.
%It shows an intriguing  phenomenon, that is, the superposition of two incoherent states may give birth to an coherent state.

However, there also exists another situation, i.e.,
\begin{equation}
 | {\Omega _2} \rangle {\rm{ = }}\frac{{\rm{1}}}{{\sqrt {\rm{2}} }} ( { | +  \rangle + | -  \rangle } ),
\end{equation}
where $| \pm  \rangle=( { | 0 \rangle \pm | 1 \rangle } )/\sqrt{2}$.
Each term in the right hand side of Eq. (3) is a maximally coherence state, but their superposition is incoherent.

In the light of the above two examples, it seems that there is no relation between the coherence of $ | \Omega  \rangle$ and that of $ | \Phi  \rangle $, $ | \Psi  \rangle$. However, in this paper, we find that the relations indeed exist, and  obtain the upper bound of the difference between the coherence of the superposition state and the average coherence of terms in the case that two terms have support on orthogonal subspaces; while in the case that two terms are in one subspace, we get other relations. Furthermore, we obtain the lower bound of the coherence of superposition in terms of the coherence of  two individual  terms. Our results might provide us a deeper insight into the research of quantum coherence.

The paper is organized as follows. First, using the relative entropy of coherence, we obtain the upper bound of coherence of the superposition in the case that two terms are in different orthogonal subspaces.
Second, upper bounds in the case that the two terms are chosen from the same subspace are also given.
Third, we deduce the lower bound of coherence of the superposition state generated by superposing two arbitrary states.
Finally, we make a conclusion.

\emph{The upper bound in two orthogonal subspaces.---}In this section, we study the relation between the coherence of two states have support on orthogonal subspaces and that of their superposition.

In order to solve this problem, we need the definition of the relative entropy of coherence \cite{duliang1}.
Given a particular basis, $\{|i\rangle_{i=1}^{d}\}$, the relative entropy of coherence
\begin{equation}C_{re} ( \rho  ) = S ( {\rho _{{\rm diag}}} ) - S ( \rho  ),\end{equation}
is a proper measure of coherence. Here, $S(\rho)$ is von Neumann entropy of $\rho$, $\rho$ is density operator and $\rho_{\rm diag}$ denotes the state obtained from $\rho$ by vanishing all off-diagonal elements. In the case of a pure state $|\psi\rangle$, its relative entropy of coherence can be expressed as
\begin{equation}
C_{re}(|\psi\rangle\langle\psi|)=S(|\psi\rangle_{{\rm diag}}\langle\psi|).
\end{equation}

In addition, the following two inequalities \cite{1} will be repeatedly used throughout this paper,
\begin{eqnarray}
&&\left| \alpha  \right|^2 S\left( \rho  \right) + \left| \beta  \right|^2 S \left( \sigma  \right) \le S\left( {\left| \alpha  \right|^2 \rho  + \left| \beta  \right|^2 \sigma } \right), \\
S\big(\left| \alpha\right|&&^2 \rho  + \left| \beta  \right|^2 \sigma \big) \le \left| \alpha  \right|^2 S \left( \rho  \right) + \left| \beta  \right|^2 S \left( \sigma  \right) + h\left( {\left| \alpha  \right|^2 } \right).~~~~
 \end{eqnarray}
Note that, the equality of (6) holds if and only if $\rho$ and $\sigma$ are identical, and the equality of (7) holds if and only if $\rho$ and $\sigma$ have support on orthogonal subspaces.

Using Eq. (5) and inequality (7), we obtain the following theorem.
%of space generated by $ \{ { | i \rangle } \}$.
\begin{theorem}
Let $ | \Phi  \rangle , | \Psi  \rangle$ be two states which have support on orthogonal  subspaces, and $\left| \Omega  \right\rangle {\rm{ = }}\alpha \left| \Phi  \right\rangle  + \beta \left| \Psi  \right\rangle$.
 Then, the coherence of the superposition satisfies
 \begin{eqnarray}
C_{re} \left( {\left| \Omega  \right\rangle \left\langle \Omega  \right|} \right) && {\rm{ = }}\left| \alpha  \right|^2 C_{re} \left( {\left| \Phi  \right\rangle \left\langle \Phi  \right|} \right)\nonumber\\
&&+ \left| \beta  \right|^2 C_{re} \left( {\left| \Psi\right\rangle \left\langle \Psi  \right|} \right) + h\left( {\left| \alpha  \right|^2 } \right),
\end{eqnarray}
where $\left| \alpha  \right|^2  + \left| \beta  \right|^2  = 1$ and $h ( x ) = - x\log x - (1 - x)\log (1 - x)$.
\end{theorem}
\begin{proof}
Combining with the conditions of Theorem 1 and inequality (7),
 we can easily get
 \begin{equation}
\begin{split}
&| \alpha  |^2 S ( { | \Phi  \rangle_{{\rm diag}} \langle \Phi  |}  ) + | \beta  |^2 S ( { | \Psi  \rangle _{{\rm diag}}\langle \Psi  |} ) + h ( { | \alpha  |^2 } )\\
{\rm{ = }}&S ( {| \Omega\rangle _{{\rm diag}} \langle\Omega |}  ).
\end{split}
\end{equation}
Thus, we  can get Eq. (8) in terms of Eq. (5).
\end{proof}

According to Theorem 1, it is not difficult to obtain the following inequality, i.e.,
 \begin {equation}
C_{re} \left( {\left| \Omega  \right\rangle \left\langle \Omega  \right|} \right) - \left| \alpha  \right|^2 C_{re} \left( {\left| \Phi  \right\rangle \left\langle \Phi  \right|} \right) - \left| \beta  \right|^2 C_{re} \left( {\left| \Psi  \right\rangle \left\langle \Psi  \right|} \right) \le 1.
 \end {equation}
Hence, the upper bound of the maximum increase of coherence is 1, and it is independent of the dimension.

\emph{The upper bounds in one space.---}
In this section, we study the relationship between the coherence of two terms in one subspace and that of their superposition state.
For two terms, we study orthogonal and non-orthogonal cases respectively.

\emph{1. Orthogonal states.}
Combining the coherence of the pure state with inequalities (6) and (7), we prove the following theorem.

\begin{theorem}
 Given two orthogonal states $ | \Phi  \rangle , | \Psi  \rangle$. Let $\left| \Omega  \right\rangle {\rm{ = }}\alpha \left| \Phi  \right\rangle  + \beta \left| \Psi  \right\rangle$ and $ | \alpha  |^2 + | \beta  |^2 = 1$.
Then, the coherence of the superposition satisfies
 \begin{eqnarray}
C_{re} \left( {\left| \Omega  \right\rangle \left\langle \Omega  \right|} \right) \le&& {\rm{2}}\Big[\left| \alpha  \right|^2 C_{re} \left( {\left| \Phi  \right\rangle \left\langle \Phi  \right|} \right) + \left| \beta  \right|^2 C_{re} \left( {\left| \Psi  \right\rangle \left\langle \Psi  \right|} \right)\nonumber\\
&&+ h\left( {\left| \alpha  \right|^2 } \right)\Big].
\end{eqnarray}
\end{theorem}

\begin{proof}
For a fixed basis $\{|i\rangle_{i=1}^{d}\}$, $|\Phi\rangle $ and $|\Psi \rangle$ can be written as  $\left| \Phi  \right\rangle  = \sum\limits_{i = 1}^d {} a_i \left| i \right\rangle ,\left| \Psi  \right\rangle  = \sum\limits_{i = 1}^d {} b_i \left| i \right\rangle$.
Let $\left| \Gamma \right\rangle {\rm{ = }}\alpha \left| \Phi  \right\rangle  - \beta \left| \Psi  \right\rangle$.
Then we have
\begin{equation}
\begin{split}
\left| \Omega  \right\rangle _{{\rm diag}}\left\langle \Omega  \right| = \sum\limits_{i = 1}^d {} \omega _i \left| i \right\rangle\left\langle i \right|,\\
\left| \Gamma  \right\rangle _{{\rm diag}}\left\langle \Gamma  \right|= \sum\limits_{i = 1}^d {} \nu _i \left| i \right\rangle\left\langle i \right|,
\end{split}
\end{equation}
where

$\omega _i  = \left| \alpha  \right|^2 \left| {a_i } \right|^2  + \left| \beta  \right|^2 \left| {b_i } \right|^2  + \alpha ^* \beta a_i^* b_i  + \alpha \beta ^* a_i b_i^* $,

$\nu _i  = \left| \alpha  \right|^2 \left| {a_i } \right|^2  + \left| \beta  \right|^2 \left| {b_i } \right|^2  - \alpha ^* \beta a_i^* b_i  - \alpha \beta ^* a_i b_i^*$.\\
Thus, we can get the following equation
\begin{eqnarray}
&&\frac{1}{2}\left| \Omega  \right\rangle _{{\rm diag}}\left\langle \Omega  \right| + \frac{1}{2}\left| \Gamma  \right\rangle _{{\rm diag}}\left\langle \Gamma  \right|
\nonumber\\
%=&& \sum\limits_{i = 1}^d {} \left( {\left| \alpha  \right|^2 \left| {a_i } \right|^2  + \left| \beta  \right|^2 \left| {b_i } \right|^2 } \right)\left| i \right\rangle \left\langle i \right|,\nonumber\\
=&& | \alpha  |^2 | \Phi  \rangle _{{\rm diag}}\langle \Phi  | + | \beta  |^2 | \Psi  \rangle _{{\rm diag}}\langle \Psi  |.
\end{eqnarray}
By using the inequalities (6) and (7), the following inequality can be deduced
\begin{equation}
\begin{split}
&\frac{1}{2}S(\left| \Omega  \right\rangle _{{\rm diag}}\left\langle \Omega  \right|) +\frac{1}{2}S(\left| \Gamma  \right\rangle _{{\rm diag}}\left\langle \Gamma  \right|)\nonumber\\
\le & S(\frac{1}{2}\left| \Omega  \right\rangle _{{\rm diag}}\left\langle \Omega  \right| + \frac{1}{2}\left| \Gamma  \right\rangle _{{\rm diag}}\left\langle \Gamma  \right|)\nonumber\\
\le &| \alpha |^2 S (| \Phi \rangle_{{\rm diag}} \langle \Phi |) + | \beta|^2 S( | \Psi  \rangle_{{\rm diag}} \langle \Psi |) + h( | \alpha  |^2 ).
\end{split}
\end{equation}
Since $S(\left| \Gamma  \right\rangle _{{\rm diag}}\left\langle \Gamma  \right|)\ge 0$, we have
\begin{eqnarray}
&& | \alpha  |^2 S  ( { | \Phi  \rangle_{{\rm diag}} \langle \Phi  |} ) + | \beta  |^2 S  ( { | \Psi  \rangle _{{\rm diag}}\langle \Psi  |} ) + h ( { | \alpha  |^2 } )\nonumber \\
\ge&&\frac{1}{2}S \left( {\left| \Omega  \right\rangle_{{\rm diag}} \left\langle \Omega  \right|} \right).
\end{eqnarray}
According to Eq. (5), we can obtain the inequality (11).
\end{proof}

Hence, we obtain the upper bound of the coherence of superposition in terms of two orthogonal terms being superposed. Though the relation is different from the one in Theorem 1, but we could consider Theorem 1 as the special situation of Theorem 2.
Next, we consider the general situation which can include Theorem 2. That is, the arbitrary state, and it can be divided into two cases: orthogonal and non-orthogonal. Since we have obtained the result of two orthogonal states, in the following we only consider the case of two non-orthogonal states.

\emph{2. Non-orthogonal states.}
Now, we consider the non-orthogonal case. Let $ | \Phi  \rangle , | \Psi  \rangle$ be two common normalized states, which we are superposing are non-orthogonal. In this case, we can obtain the  following theorem.

\begin{theorem}
Let two states $ | \Phi  \rangle , | \Psi  \rangle$ be normalized but non-orthogonal. Let $| {{\rm T}_1 } \rangle = \frac{{\alpha | \Phi  \rangle + \beta | \Psi  \rangle }}{{ \| {\alpha | \Phi  \rangle + \beta | \Psi  \rangle } \|}}$ and $ | \alpha  |^2 + | \beta  |^2 = 1$.
Then, the coherence of the superposition satisfies
 \begin{eqnarray}
\left\| {\alpha \left| \Phi  \right\rangle  + \beta \left| \Psi  \right\rangle } \right\|^2&& C_{re} \left( {\left| {{\rm T}_1 } \right\rangle \left\langle {{\rm T}_1 } \right|} \right)
\le {\rm{2}}\Big[ \left| \alpha  \right|^2 C_{re} \left( {\left| \Phi  \right\rangle \left\langle \Phi  \right|} \right) \nonumber\\
&&+ \left| \beta  \right|^2 C_{re} \left( {\left| \Psi  \right\rangle \left\langle \Psi  \right|} \right) + h\left( {\left| \alpha  \right|^2 } \right) \Big].
\end{eqnarray}\nonumber
\end{theorem}

\begin{proof}
For a fixed basis $\{|i\rangle_{i=1}^{d}\}$, $|\Phi\rangle $ and $|\Psi \rangle$ can be written as  $\left| \Phi  \right\rangle  = \sum\limits_{i = 1}^d {} a_i \left| i \right\rangle ,\left| \Psi  \right\rangle  = \sum\limits_{i = 1}^d {} b_i \left| i \right\rangle$. So it is not difficult to get
\begin{equation}
 \| {\alpha | \Phi  \rangle + \beta | \Psi  \rangle } \|^2 + \| {\alpha | \Phi  \rangle - \beta | \Psi  \rangle } \|^2 = 2.
\end{equation}
We consider the following two superposition states
\begin{eqnarray}
| {{\rm T}_1 } \rangle = \frac{{\alpha | \Phi  \rangle + \beta | \Psi  \rangle }}{{ \| {\alpha | \Phi  \rangle + \beta | \Psi  \rangle } \|}},  | {{\rm T}_2 } \rangle = \frac{{\alpha | \Phi  \rangle - \beta | \Psi  \rangle }}{{ \| {\alpha | \Phi  \rangle - \beta | \Psi  \rangle } \|}},
\end{eqnarray}
which satisfy normalization.
Then we have
\begin{eqnarray}
&&\frac{{\left\| {\alpha \left| \Phi  \right\rangle  + \beta \left| \Psi  \right\rangle } \right\|^2 }}{2}\left| {{\rm T}_1 } \right\rangle _{{\rm diag}} \left\langle {{\rm T}_1 } \right| \nonumber\\
+&& \frac{{\left\| {\alpha \left| \Phi  \right\rangle  - \beta \left| \Psi  \right\rangle } \right\|^2 }}{2}\left| {{\rm T}_2 } \right\rangle_{{\rm diag}}  \left\langle {{\rm T}_2 } \right|\nonumber\\
%=&& \sum\limits_{i = 1}^d {} \left( {\left| \alpha  \right|^2 \left| {a_i } \right|^2  + \left| \beta  \right|^2 \left| {b_i } \right|^2 } \right)\left| i \right\rangle \left\langle i \right|\nonumber\\
=&&\left| \alpha  \right|^2 \left| \Phi  \right\rangle_{{\rm diag}}  \left\langle \Phi  \right| + \left| \beta  \right|^2 \left| \Psi  \right\rangle_{{\rm diag}}  \left\langle \Psi  \right|.
\end{eqnarray}
Hence, using the inequalities (6) and (7), we can obtain the following inequality
\begin{eqnarray}
&&\frac{{ \| {\alpha | \Phi  \rangle + \beta | \Psi  \rangle } \|^2 }}{2}S  ( { | {{\rm T}_1 } \rangle _{{\rm diag}} \langle {{\rm T}_1 } |} )  \nonumber \\
&&+ \frac{{ \| {\alpha | \Phi  \rangle - \beta | \Psi  \rangle } \|^2 }}{2}S ( { | {{\rm T}_2 } \rangle_{{\rm diag}}  \langle {{\rm T}_2 } |} ) \nonumber\\
&&\le  | \alpha  |^2 S ( { | \Phi  \rangle_{{\rm diag}}  \langle \Phi  |} ) + | \beta  |^2 S ( { | \Psi  \rangle_{{\rm diag}}  \langle \Psi |} ) + h ( { | \alpha  |^2 } ).~~
\end{eqnarray}
Due to the non-negativity of von Neumann entropy, we can obtain the following inequality

\begin{eqnarray}
&&\frac{{ \| {\alpha | \Phi  \rangle + \beta | \Psi  \rangle } \|^2 }}{2}S ( { | {{\rm T}_1 } \rangle _{{\rm diag}} \langle {{\rm T}_1 } |} )   \nonumber\\
&&\le | \alpha  |^2 S ( { | \Phi  \rangle_{{\rm diag}}  \langle \Phi  |} )+ | \beta  |^2 S  ( { | \Psi  \rangle_{{\rm diag}} \langle \Psi  |} ) + h ( { | \alpha  |^2 } ).~~
\end{eqnarray}
Because the von Neumann entropy of $\left| \Phi  \right\rangle$, $\left| \Psi  \right\rangle$
are both zero, the inequality (15) can be derived.
\end{proof}
 \hspace{1cm}

\emph{The lower bound of coherence of superposition.---}
Up to now, upper bounds of coherence of superposition have been studied. Naturally, we wonder if there is a lower bound of coherence of superpositions. By the following theorem, we give an affirmative answer.
\begin{widetext}
\begin{theorem}
Let two states $ | \Phi  \rangle , | \Psi  \rangle$ be normalized but otherwise arbitrary, $\left| {{\rm T}_1 } \right\rangle  = \frac{{\alpha \left| \Phi  \right\rangle  + \beta \left| \Psi  \right\rangle }}{s}$ and $ | \alpha  |^2 + | \beta  |^2 = 1$.
Then, the coherence of the superposition satisfies
\begin{eqnarray}
s^2 C_{re} \left( {\left| {{\rm T}_1 } \right\rangle \left\langle {{\rm T}_1 } \right|} \right)&& \ge \max \Bigg\{\frac{{\left| \alpha  \right|^2 }}{2}C_{re} \left( {\left| \Phi  \right\rangle \left\langle \Phi  \right|} \right) - \left| \beta  \right|^2 C_{re} \left( {\left| \Psi  \right\rangle \left\langle \Psi  \right|} \right)- \left( {s^2  + \left| \beta  \right|^2 } \right)h\left( {\frac{{\left| \beta  \right|^2 }}{{s^2  + \left| \beta  \right|^2 }}} \right),\nonumber\\
&&\frac{{\left| \beta  \right|^2 }}{2}C_{re} \left( {\left| \Psi  \right\rangle \left\langle \Psi  \right|} \right)
- \left| \alpha  \right|^2 C_{re} \left( {\left| \Phi  \right\rangle \left\langle \Phi  \right|} \right) - \left( {s^2  + \left| \alpha  \right|^2 } \right)h\left( {\frac{{\left| \alpha  \right|^2 }}{{s^2  + \left| \alpha  \right|^2 }}} \right) \Bigg\},
\end{eqnarray}
where $\left\| {\alpha \left| \Phi  \right\rangle  + \beta \left| \Psi  \right\rangle } \right\| ={s}.$
\end{theorem}
\end{widetext}
\begin{proof}
Using $ | {{\rm T}_1 } \rangle$, $ | \Phi  \rangle$ and $| \Psi  \rangle$ can be separately rewritten as
\begin{equation}
\frac{{\alpha | \Phi  \rangle }}{{\sqrt {s^2 + | \beta  |^2 } }}{\rm{ = }}\frac{{{s} | {{\rm T}_1 } \rangle }}{{\sqrt {s^2 + | \beta  |^2 } }} - \frac{\beta }{{\sqrt {s^2 + | \beta  |^2 } }} | \Psi  \rangle,
\end{equation}
and
\begin{equation}
\frac{{\beta | \Psi  \rangle }}{{\sqrt {s^2 + | \alpha   |^2 } }}{\rm{ = }}\frac{{{s} | {{\rm T}_1 } \rangle }}{{\sqrt {s^2 + | \alpha   |^2 } }} - \frac{\alpha  }{{\sqrt {s^2 + | \alpha  |^2 } }} | \Phi \rangle.
\end{equation}
According to Theorem 2 and Theorem 3, we have
\begin{eqnarray}
&&\frac{{\left| \alpha  \right|^2 }}{{s^2  + \left| \beta  \right|^2 }} C_{re} \left( {\left| \Phi  \right\rangle \left\langle \Phi  \right|} \right) \le 2\Bigg[\frac{{{s}^{\rm{2}} }}{{s^2 + \left| \beta  \right|^2 }}C_{re} \left( {\left| {{\rm T}_1 } \right\rangle \left\langle {{\rm T}_1 } \right|} \right)  \nonumber\\
&&~~~~~+ \frac{{\left| \beta  \right|^2 }}{{s^2  + \left| \beta  \right|^2 }}C_{re} \left( {\left| \Psi  \right\rangle \left\langle \Psi  \right|} \right) + h\left( {\frac{{\left| \beta  \right|^2 }}{{s^2  + \left| \beta  \right|^2 }}} \right)\Bigg].
\end{eqnarray}
Furthermore, through the transpose and  simplify, we obtain the following inequality

\begin{eqnarray}
s^2 C_{re} \left( {\left| {{\rm T}_1 } \right\rangle \left\langle {{\rm T}_1 } \right|} \right) &&\ge \frac{{\left| \alpha  \right|^2 }}{2}C_{re} \left( {\left| \Phi  \right\rangle \left\langle \Phi  \right|} \right) - \left| \beta  \right|^2 C_{re} \left( {\left| \Psi  \right\rangle \left\langle \Psi  \right|} \right) \nonumber\\
&&- \left( {s^2  + \left| \beta  \right|^2 } \right)h\left( {\frac{{\left| \beta  \right|^2 }}{{s^2  + \left| \beta  \right|^2 }}} \right).
\end{eqnarray}
On the other hand,  we can deduce the other inequality by the same method. That is

\begin{eqnarray}
s^2 C_{re} \left( {\left| {{\rm T}_1 } \right\rangle \left\langle {{\rm T}_1 } \right|} \right) &&\ge \frac{{\left| \beta  \right|^2 }}{2}C_{re} \left( {\left| \Psi  \right\rangle \left\langle \Psi  \right|} \right) - \left| \alpha  \right|^2 C_{re} \left( {\left| \Phi  \right\rangle \left\langle \Phi  \right|} \right) \nonumber\\
&&- \left( {s^2  + \left| \alpha  \right|^2 } \right)h\left( {\frac{{\left| \alpha  \right|^2 }}{{s^2  + \left| \alpha  \right|^2 }}} \right).
\end{eqnarray}
Therefore, both of the inequalities (25) and (26) should be satisfied, and now we have finished the proof of Theorem 4.
\end{proof}

%As an application of Theorem 2 and Theorem 3, we consider the superposition
%\begin{equation}
%\left| {{\rm T}_1 } \right\rangle  = \alpha \frac{{\left\langle {x\left| \Phi  \right\rangle } \right.}}{{\left| {\left\langle {x\left| \Phi  \right\rangle } \right.} \right|}}\left| \Psi  \right\rangle  + \beta \frac{{\left\langle {x\left| \Psi  \right\rangle } \right.}}{{\left| {\left\langle {x\left| \Psi  \right\rangle } \right.} \right|}}\left| \Phi  \right\rangle,
%\end {equation}
%where $|x\rangle$ is a fixed pure state, and satisfying $\langle x|\Psi\rangle\neq0$, $\langle x|\Phi\rangle\neq0$.
%Eq. (27) comes from Ref. \cite{22}, wherein the problem how to create the superposition states of unknown states is studied, and  their  proposal can be realized in experiment \cite{23}.
%According to Theorem 3, we can obtain that the ratio of the coherence of $|T_{1}\rangle$ and $ \left| \alpha  \right|^2 C_{re} \left( {\left| \Phi  \right\rangle \left\langle \Phi  \right|} \right) + \left| \beta  \right|^2 C_{re} \left( {\left| \Psi  \right\rangle \left\langle \Psi  \right|} \right) + h\left( {\left| \alpha  \right|^2 } \right)$ is smaller than $2/(1+2Re(\alpha^{*}\beta\frac{{{\rm{tr}}(| x \rangle \langle x | \Psi\rangle \langle \Psi  | \Phi  \rangle \langle \Phi |)}}{{\left| {\left\langle x \right|\left. \Phi  \right\rangle \left\langle x \right|\left. \Psi  \right\rangle } \right|}}))$.

\emph{Conclusion.---}
We give upper bounds of the coherence of the superpositions in terms of the coherence of two terms being supposed. Meanwhile, we obtain the lower bound for the coherence of the superpositions.
In addition, the case of multiple terms in superposition can be directly built using our methods.
Our results might give us a deeper understanding of quantum coherence and attract us to research other novel relations between coherence and entanglement.
Noting that a recent work have investigated the creation of a superposition of unknown quantum states \cite{22}, we hope our results can also be used to bound the coherence of the created superposition state.
\section*{Acknowledgements}
This work is supported by NSFC (Grant Nos. 61272057, 61572081).

%The extraction of work from quantum coherence.
\end{document}